\documentclass[final, nomarks]{dmtcs-episciences}
\usepackage[utf8]{inputenc}
\usepackage[round]{natbib}

\usepackage{amsmath} 
\usepackage{amsfonts} 
\usepackage{graphicx}
\usepackage{bm}
\usepackage{tikz}
\usepackage{amsthm}
\usepackage{tabularray}
\usepackage{rotating}
\graphicspath{ {./Math/} }
\usepackage{tikz-qtree}
\usetikzlibrary{automata, positioning, arrows}
\tikzset{
-, 
>=stealth, 
node distance=3cm, 
every state/.style={thick, fill=gray!10}, 
initial text=$ $, 
}

\newtheorem{theorem}{Theorem}[section]

\newtheorem{lemma}[theorem]{Lemma}
\theoremstyle{definition}
\newtheorem{definition}[theorem]{Definition}
\newcommand{\fzip}{\mbox{{\sf fzip}}}
\newcommand{\zip}{\mbox{{\sf zip}}}

\title{A Diamond Structure in the Transducer Hierarchy}
\author{Noah Kaufmann}	
\affiliation{Department of Mathematics, Indiana University, Bloomington, USA}

\begin{document}
\publicationdata
{vol. 27:3}
{2025}
{21}
{10.46298/dmtcs.8704}
{2021-11-13; 2021-11-13; 2022-08-05; 2023-01-18}
{2025-10-26}
\maketitle

\begin{abstract}

\bigskip
We answer an open question in the theory of transducer degrees on the existence of a diamond structure in the transducer hierarchy. Transducer degrees are the equivalence classes formed by word transformations which can be realized by a finite state transducer, which form an order based on which words can be transformed into other words. We provide a construction which proves the existence of a diamond structure, while also introducing a new function on streams which may be useful for proving more results about the transducer hierarchy.
\end{abstract}

\section{Introduction}

Finite state transducers (FSTs) are ubiquitous in computer science, and infinite streams are also common in many fields. Yet there are very few results on how to transform an arbitrary stream into another stream with an FST. We define a stream $\sigma$ as being above another stream $\tau$ if some FST $T$ can transduce $\sigma$ into $\tau$, with $\sigma$ and $\tau$ being the same degree if they can both be transduced into each other. The structure of these degrees, called the transducer hierarchy, has many parallels with Turing degrees. The main results that have been done thus far mostly deal with streams determined by polynomials, and in this paper we will present a new result that comes from a new operation on streams. This new operation, called $\fzip$, may be useful in proving additional results, and allows us a new class of functions to consider the degrees of: piecewise polynomials. 

\section{Definitions}

We will give some preliminary definitions with the goal of understanding the definition of a weight product, the key operation for all the results in this paper. For more definitions and background in this area, see \cite{1,2,3,4,5}. We begin by setting \textbf{2} = \{0,1\}, which we will use as our input and output alphabet for all of our transducers. We use $\textbf{2}^\infty$ to denote the set of all finite and infinite streams over $\textbf{2}$. We will focus only on finite state transducers of the following form: 

\begin{definition} A \textit{finite-state transducer} is a tuple $T = \langle Q, q_0,\delta,\lambda \rangle$ where $Q$ is a finite set of states, $q_0 \in Q$ is the initial state, $\delta : Q \times \textbf{2} \rightarrow Q$ is the transition function, and $\lambda : Q \times \textbf{2} \rightarrow \textbf{2}^*$ is the output function.  
\end{definition}

Note that $\delta$ and $\lambda$ can be extended ($\delta : Q \times \textbf{2}^* \rightarrow Q,\lambda : Q \times \textbf{2}^\infty \rightarrow \textbf{2}^\infty$) as follows: \\
$$\delta(q,\epsilon) = q, \delta(q,au) = \delta(\delta(q,a),u), \textrm{ where } q \in Q, a \in \textbf{2}, u \in \textbf{2}^*$$
$$\lambda(q, \epsilon) = \epsilon, \lambda(q, au) = \lambda(q,a) \cdot \lambda(\delta(q,a),u), \textrm{ where } q \in Q, a \in \textbf{2}, u \in \textbf{2}^\infty$$ 

The above equations correspond to inputting a (possibly infinite) stream of letters into $T$. This allows us to define a function $T$ on finite or infinite strings in \textbf{2}, by saying that $T(w)$ is equal to the output of the FST $T$ after inputting $w$. Formally speaking, this means that $T(w) = \lambda(q_0, w)$. Now we can make the following definition, which is the basis for the transducer hierarchy:

\begin{definition}  Let $T$ be an FST, and let $\sigma, \tau \in \textbf{2}^{\infty}$ be infinite sequences. We say that $T$ $transduces$ $\sigma$ to $\tau$, or that $\tau$ is the $T$-$transduct$ of $\sigma$, if $T(\sigma) = \tau$. In general, for any two infinite sequences $\sigma, \tau$ we say that $\sigma \geq \tau$ if there exists some $T$ so that $T(\sigma) = \tau$.
\end{definition}

This relation $\geq$ is reflexive, and can be shown to be transitive by composition of FSTs (See Lemma 8, \cite{1}). If for some $\sigma, \tau$ we have $\sigma \geq \tau$ but not vice versa, we say $\sigma > \tau$. If we do have $\sigma \geq \tau$ and $\tau \geq \sigma$ then we say that $\sigma \equiv \tau$, and we use $[\sigma]$ to denote the equivalence class of $\sigma$. We call $[\sigma]$ the \textit{degree} of $\sigma$. 

Now that we have defined what a transducer degree is, we will focus our attention on a particular subset of streams, namely the streams which are generated by functions in the sense of the following definition.

\begin{definition}
For a function $f$ from $\mathbb{N}$ to $\mathbb{N}$, we define $\langle f \rangle$ to be the stream given by

\begin{center}$\langle f \rangle = \prod_{i=0}^{\infty} 10^{f(i)} = 10^{f(0)}10^{f(1)}10^{f(2)} \ldots$\end{center} 

We will often use $\langle f \rangle$ to mean both the stream determined by $f$, as well as the degree of that stream $[\langle f \rangle]$. We also refer to a part of the stream of the form $10^{f(i)}$ as a block. 

\end{definition}

Having defined $\langle f \rangle$ in this way, some relatively simple initial results have been obtained in \cite{2}, which we state here.

\begin{lemma} Let $f: \mathbb{N} \rightarrow \mathbb{N}, a,b \in \mathbb{N}.$ We have the following equivalences and inequalities:\\

\begin{enumerate}
\item $\langle af(n) \rangle \equiv \langle f(n) \rangle, $ for $ a > 0$\\
\item $\langle f(n + a) \rangle \equiv \langle f(n) \rangle$\\
\item $\langle f(n) + a \rangle \equiv \langle f(n) \rangle$\\
\item $\langle f(n) \rangle \geq \langle f(an) \rangle, $ for $ a > 0$\\
\item $\langle f(n) \rangle \geq \langle af(2n) + bf(2n + 1) \rangle$\\
\end{enumerate}

\end{lemma}

One interesting consequence of the third equality is that any polynomial with a positive leading coefficient can be thought of as a stream and thus associated with a transducer degree, even if some of its values happen to be negative. For instance, the polynomial $(n-2)^3$ can't directly be interpreted as a stream, since it is negative for $n = 0, 1$. However, if we take $(n-2)^3 + 8$, then this polynomial is nonnegative, and therefore corresponds to a stream (and thus a degree). So even though it's technically incorrect, it will be convenient sometimes to refer to a degree such as $\langle (n-2)^3 \rangle$, when we mean more precisely the degree $\langle (n-2)^3 + k \rangle$ for any $k \geq 8$. Similarly, the first equality allows us to refer to the degree of a function with rational coefficients, where we really mean the degree of the corresponding function multiplied by the appropriate scalar to eliminate any fractional coefficients.

Now we are ready to start defining weight products. We will not provide the full proof of the main result we need (Theorem 2.6), but a more detailed explanation can be found in \cite{2}. We begin by defining a weight.

\begin{definition} A \textit{weight} is a tuple $\alpha = \langle a_0, a_1, \ldots , a_{k-1},b \rangle \in \mathbb{Q}^{k+1}$ with each $a_i \geq 0$. If $a_i = 0$ for all $i$ then we say the weight is constant. To distinguish between weights and tuples of weights, weights will not be bolded but tuples of weights will, except potentially in cases where there is only one weight in the tuple.
\end{definition}

Given a weight $\alpha$ as above and a function $f: \mathbb{N} \rightarrow \mathbb{N}$ we can define $\alpha \cdot f$ as:

$$\alpha \cdot f = a_0f(0) + a_1f(1)+ \ldots + a_{k-1}f(k-1)+b$$

We are ready to define the weight product. Let $\bm{\alpha} = \langle \alpha_0, \alpha_1, \ldots , \alpha_{m-1} \rangle$ be a tuple of weights, with $\bm{\alpha^\prime}$ being the cyclic shift $\langle \alpha_1, \alpha_2, \ldots , \alpha_{m-1}, \alpha_0 \rangle$.

Then the weight product of $\bm{\alpha}$ with $f$, written as $\bm{\alpha} \otimes f$, is defined in the following way:

$$(\bm{\alpha} \otimes f)(0) = \alpha_0 \cdot f$$
$$(\bm{\alpha} \otimes f)(n+1) = (\bm{\alpha^\prime} \otimes S^{|\alpha_0|-1}(f))(n)$$ 

Here $S^k(f)(n) = f(n+k)$ and $|\alpha_0|$ indicates the length of the tuple $\alpha_0$. We call a weight product $natural$ if $\bm{\alpha} \otimes f(n) \in \mathbb{N}$ for all $n$. Note that since $\bm{\alpha}$ is a finite tuple of finite tuples in $\mathbb{Q}$, we can take the LCM of all of the denominators and multiply through to make the product natural. Since this does not change the degree of the resulting function (by Lemma 2.4), from now on we will assume that all weight products are natural. We also define the length of a tuple of weights to be $||\bm{\alpha}|| = \sum_{i=0}^{m-1} (|\alpha_i| - 1)$. (By $|\alpha_i|$ we mean simply the number of elements in that weight.) Finally, we note that in the case where the tuple of weights only contains one weight, i.e. $\bm{\alpha} = (\alpha_0)$, we will often use $\alpha_0 \otimes f$ to denote $\bm{\alpha} \otimes f$.   

The following image provides a more intuitive picture of how the weight product works, by showing pictorially how to compute the weight product of the tuple of weights $\bm{\alpha} = \langle \alpha_0, \alpha_1 \rangle$ with an arbitrary function $f(n)$, where $\alpha_0 = \langle 2,4,6,8 \rangle, \alpha_1 = \langle 1,7,4 \rangle$:

\begin{enumerate}
\item[] 
\begin{tikzpicture}[grow=up]
\Tree [.{$(\bm{\alpha} \otimes f)(0) = 2f(0)+4f(1)+6f(2)+8$} [.{$\times 6$} $f(2)$ ] [.{$\times 4$} $f(1)$ ] [.{$\times 2$} $f(0)$ ] ]
\end{tikzpicture}
\hskip 0.1in
\begin{tikzpicture}[grow=up]
\Tree [.{$(\bm{\alpha} \otimes f)(1) = f(3)+7f(4)+4$} [.{$\times 7$} $f(4)$ ] [.{$\times 1$} $f(3)$ ] ]
\end{tikzpicture}
\end{enumerate}

\begin{enumerate}
\item[] 
\begin{tikzpicture}[grow=up]
\Tree [.{$(\bm{\alpha} \otimes f)(2) = 2f(5)+4f(6)+6f(7)+8$} [.{$\times 6$} $f(7)$ ] [.{$\times 4$} $f(6)$ ] [.{$\times 2$} $f(5)$ ] ]
\end{tikzpicture}
\hskip 0.1in
\begin{tikzpicture}[grow=up]
\Tree [.{$(\bm{\alpha} \otimes f)(3) = f(8)+7f(9)+4$} [.{$\times 7$} $f(9)$ ] [.{$\times 1$} $f(8)$ ] ]
\end{tikzpicture}
\end{enumerate}

This also provides us with a better notion of what the length $||\bm{\alpha}||$ represents. For this $\bm{\alpha}$, we have $||\bm{\alpha}|| = (4-1) + (3-1) = 5$, which is exactly how many values of $f$ we go through after applying every weight. 

The key property of weight products is that they can replace transducers when used on a certain class of streams known as \textit{spiralling functions}. For the purposes of this paper, we are interested in two subclasses of spiralling functions: polynomials and piecewise polynomials. In this paper, by ``piecewise polynomial" we refer to a slightly different definition than the standard one. Here, we define a piecewise polynomial to be a function which, given a finite partition of $\mathbb{N}$, agrees with a polynomial on each element of this partition. In particular, the partitions we use will be residue classes modulo some $n \in \mathbb{N}$. We will refer to a polynomial defined on a particular element of the partition as a ``piece" of the piecewise polynomial.

Now that we have defined weight products as well as the class of functions we wish to apply them to, we can proceed to state the main result that we need for the rest of the paper. The full proof of this result can be found in (Theorem 21, \cite{3}).

\begin{theorem} Let $f,g: \mathbb{N} \rightarrow \mathbb{N}$ be (possibly piecewise) polynomials. Then $\langle g \rangle \geq \langle f \rangle$ if and only if there exists a tuple of weights $\bm{\alpha}$ and integers $n_0,m_0$ such that $S^{n_0}(f) = \bm{\alpha} \otimes S^{m_0}(g)$.
\end{theorem}

This theorem tells us that if we want to show that one polynomial degree is above another, we can consider weight products, rather than trying to figure out a transducer directly. The following theorem will give us a useful result for comparing polynomial degrees with non-polynomial degrees.

\begin{theorem}
Let $f : \mathbb{N} \rightarrow \mathbb{N}$ be a (possibly piecewise) polynomial, and $\sigma \in 2^\mathbb{N}$. Then $\langle f \rangle \geq \sigma$ if and
only if $\sigma \equiv \langle \alpha \otimes S^{n_0}(f) \rangle$ for some integer $n_0 \geq 0$, and a tuple of weights $\alpha$.
\end{theorem}

These two theorems justify the use of weight products as a replacement for transducers. This greatly simplifies things, since we no longer need to consider how to transduce a stream by any transducer. Instead, we can consider weight products, which are much easier to work with. Having defined weight products, we now turn to a new operation on streams: $\fzip$.

\section{Basic Results}

Given two streams $\sigma = \sigma_0\sigma_1\sigma_2 \cdots$ and $\tau = \tau_0\tau_1\tau_2 \cdots$, one natural operation to define on them is the $\zip$ operation, also called merge in some contexts. This operation is defined by interleaving the two streams, that is, alternating elements from each stream. The formal definition is as follows:

\begin{definition} The $i$th term of the stream $\zip(\sigma,\tau)$ is given by the following equation:
\end{definition}

\begin{equation*}
\zip(\sigma,\tau)_i=\begin{cases}
          \sigma_{\frac{i}{2}} \quad &\text{i even}\\
          \tau_{\frac{(i-1)}{2}} \quad &\text{i odd}\\
     \end{cases}
\end{equation*}\\ 

Intuitively, this is very easy to understand: the stream will simply start with the first element of $\sigma$, followed by the first element of $\tau$, then the second element of $\sigma$, and so on. This operation has many applications, and in the field of transducer degrees one property is immediately obvious: $\zip(\sigma,\tau) \geq \sigma$ and $\zip(\sigma,\tau) \geq \tau$. The fact that an FST can be constructed for each of these inequalities is easy to prove, as shown in the following image:

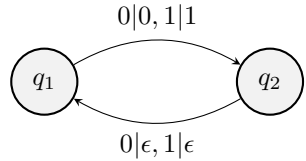
\begin{figure}[ht]
\centering
\begin{tikzpicture}[->]
\node[state] (q2) {$q_1$};
\node[state, right of=q2] (q3) {$q_2$};
\draw 
(q2) edge[bend left, above] node{$0|0, 1|1$} (q3)
(q3) edge[bend left, below] node{$0|\epsilon ,1|\epsilon$} (q2);
\end{tikzpicture}
\caption{Making $q_1$ the initial state proves that $\zip(\sigma,\tau) \geq \sigma$, and making $q_2$ the initial state proves that $\zip(\sigma,\tau) \geq \tau$.}
\end{figure}

However, when working with polynomial degrees, $\zip$ is not quite the right operation. The main problem with $\zip$ is that if we apply it to two polynomial streams, it's not clear that the result is a polynomial stream, and it's also not clear how to define it in terms of the original streams. For example, suppose we wanted to take $\sigma = \zip(\langle n \rangle, \langle n^2 \rangle)$. The beginning of this stream is given by $111100110000100001001...$, which corresponds to a function $f(n)$ with values as shown in the table below.

\begin{center}
  \begin{tblr}{hlines, vlines, columns={co=1}} 
    n & 0 & 1 & 2 & 3 & 4 & 5 & 6 & 7    \\
    $f(n)$ & 0 & 0 & 0 & 2 & 0 & 4 & 4 & 2   \\
  \end{tblr}
\end{center} 

It's not at all clear based on the above table that this function is even a polynomial, and since it has at least four zeros it must be at least a quartic if it is. So $\zip$ is difficult to work with if we are dealing with polynomial streams. Fortunately, there is a similar operation which is very useful for our purposes. \\

\begin{definition}
Let $f$ and $g$ be functions from $\mathbb{N} \rightarrow \mathbb{N}$. \\ 

We define $\fzip(f,g)$ to be the function

\begin{equation*}
\fzip(f,g)(n)=\begin{cases}
          f(\frac{n}{2}) \quad &\text{n even}\\
          g(\frac{n-1}{2}) \quad &\text{n odd}\\
     \end{cases}
\end{equation*}
\end{definition}

Therefore $\langle \fzip(f,g) \rangle$ is

\begin{center}$\langle \fzip(f,g) \rangle = \prod_{i=0}^{\infty} 10^{\fzip(f,g)(i)} = 10^{f(0)}10^{g(0)}10^{f(1)}10^{g(1)} \ldots$\end{center} 

And now the advantage of $\fzip$ becomes clear: instead of interleaving individual letters, we can interleave the blocks of zeros determined by our functions. Just as with $\zip$, it's immediately clear that this stream lies above the streams corresponding to $f$ and $g$, in this case with a slightly more complex transducer. 

\begin{figure}[ht]
\centering
\begin{tikzpicture}[->]
\node[state] (q2) {$q_1$};
\node[state, right of=q2] (q3) {$q_2$};
\draw (q2) edge[loop, above] node{$0|0$} (q2)
(q2) edge[bend left, above] node{$1|\epsilon$} (q3)
(q3) edge[loop, above] node{$0|\epsilon$} (q3)
(q3) edge[bend left, below] node{$1|1$} (q2);
\end{tikzpicture}
\caption{Similar to before, making $q_2$ the initial state proves that $\fzip(f,g) \geq f$, and making $q_1$ the initial state proves that $\fzip(f,g) \geq g$.}
\end{figure}
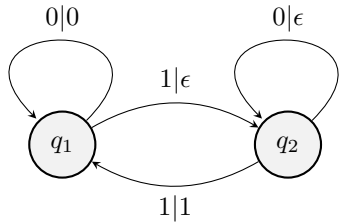

Due to this similarity, we do actually have one identity relating $\zip$ and $\fzip$, in the case where we are zipping a functional stream with itself.

\begin{lemma}
For all functions $f$ from $\mathbb{N}$ to $\mathbb{N}$, $\zip(\langle f \rangle,\langle f \rangle)= \langle \fzip(0,2f) \rangle$, and this stream is in the same degree as $\langle f \rangle$.
\end{lemma}

\begin{proof}
We know that \begin{center}$\langle f \rangle = \prod_{i=0}^{\infty} 10^{f(i)} = 10^{f(0)}10^{f(1)}10^{f(2)} \ldots$\end{center}

Zipping this stream with itself produces the following stream: \\

$\zip(\langle f \rangle, \langle f \rangle) = 110^{2f(0)}110^{2f(1)}110^{2f(2)} \ldots$ \\

Because of the pattern of $11$'s, this stream corresponds to a function which alternates between being identically zero and two times the original function, i.e. $\fzip(0, 2f)$. Since it can be written as a $\zip$ of $\langle f \rangle$, we know that its degree is greater than or equal to the degree of $\langle f \rangle$, and applying a transducer which doubles every letter to $\langle f \rangle$ gives us the reverse inequality. Therefore these streams are in the same degree.
\end{proof}

A less obvious property is that $\fzip$ is not symmetric in general, so $\langle \fzip(f,g) \rangle$ need not equal $\langle \fzip(g,f) \rangle$. This was proven inadvertently in \cite{5}, where a careful reading of the main proof in terms of $\fzip$ gives us a specific example of two functions $f,g$ which satisfy $\langle \fzip(f,g) \rangle \neq \langle \fzip(g,f) \rangle$. However, there are two classes of functions which do allow $\fzip$ to be symmetric, as detailed in the following two lemmas.

\begin{lemma}
For all functions $f,g: \mathbb{N} \rightarrow \mathbb{N}$, $\langle \fzip(f(n),g(n)) \rangle \equiv \langle \fzip(g(n), f(n+1)) \rangle$.
\end{lemma}

\begin{proof}
This is easiest to understand in the form of a ``proof by table". Consider the following table of values for $h(n) = \fzip(f(n),g(n))$:

\begin{center}
  \begin{tblr}{hlines, vlines, columns={co=1}} 
    n & 0 & 1 & 2 & 3 & 4 & 5 & 6 & 7    \\
    $h(n)$ & f(0) & g(0) & f(1) & g(1) & f(2) & g(2) & f(3) & g(3)   \\
  \end{tblr}
\end{center} 

Now if we delete the first block of the stream $\fzip(f(n),g(n))$, which an FST can do, this corresponds to the following table:

\begin{center}
  \begin{tblr}{hlines, vlines, columns={co=1}} 
    n & 0 & 1 & 2 & 3 & 4 & 5 & 6 & 7    \\
    $h'(n)$ & g(0) & f(1) & g(1) & f(2) & g(2) & f(3) & g(3) & f(4)   \\
  \end{tblr}
\end{center}

Therefore $h'(n) = g(\frac{n}{2})$ for even $n$, and $h'(n) = f(\frac{n+1}{2})$ for odd $n$. So by definition, $h'(n) = \fzip(g(n), f(n+1))$. This shows that $\langle \fzip(f(n), g(n)) \rangle \geq \langle \fzip(g(n),f(n+1)) \rangle$. The reverse inequality is very similar, except that we need to add the block $10^{f(0)}$ to the beginning of the stream $\fzip(g(n), f(n+1))$, which again can be done by an FST. Therefore $\langle \fzip(f(n),g(n)) \rangle \equiv \langle \fzip(g(n), f(n+1)) \rangle$.
\end{proof}

With this lemma, we can now show that there are two classes of functions for which $\fzip$ is symmetric: linear and exponential functions.

\begin{lemma}
Suppose $f(n) = an+b$ or $f(n) = ab^n$ for some $a,b \in \mathbb{N}$, and let $g$ be any function from $\mathbb{N}$ to $\mathbb{N}$. Then $\langle \fzip(f(n),g(n)) \rangle \equiv \langle \fzip(g(n),f(n)) \rangle$.
\end{lemma}

\begin{proof}
We have two distinct cases: $f(n) = an+b$ and $f(n) = ab^n$.

Case 1: $f(n) = an+b$

First, apply the previous lemma to $\fzip(f(n),g(n))$ to say that $\langle \fzip(f(n),g(n)) \rangle \equiv \langle \fzip(g(n), f(n+1)) \rangle$. The function $\fzip(g(n), f(n+1))$ has the form

\begin{equation*}
\fzip(g(n), f(n+1))=\begin{cases}
          g(\frac{n}{2}) \quad &\text{n even}\\
          a(\frac{n+1}{2})+b \quad &\text{n odd}\\
     \end{cases}
\end{equation*}

We can rewrite this as

\begin{equation*}
\fzip(g(n), f(n+1))=\begin{cases}
          g(\frac{n}{2}) \quad &\text{n even}\\
          a(\frac{n-1}{2})+a+b \quad &\text{n odd}\\
     \end{cases}
\end{equation*}

We can now delete $a$ zeroes from each odd block in the corresponding stream. Since this is a reversible operation, the stream corresponding to $\fzip(g(n),f(n+1))$ is in the same degree as the function corresponding to the stream below:

\begin{equation*}
h(n)=\begin{cases}
          g(\frac{n}{2}) \quad &\text{n even}\\
          a(\frac{n-1}{2})+b \quad &\text{n odd}\\
     \end{cases}
\end{equation*}

But this new function $h(n)$ is just $\fzip(g(n),f(n))$. Therefore if $f(n) = an+b$, $\langle \fzip(f(n),g(n)) \rangle \equiv \langle \fzip(g(n),f(n)) \rangle$. \\

Case 2: $f(n) = ab^n$

This case will proceed almost identically to the first case.

As before, we start by applying the previous lemma to $\fzip(f(n),g(n))$ to say that $\langle \fzip(f(n),g(n)) \rangle \equiv \langle \fzip(g(n), f(n+1)) \rangle$. The function $\fzip(g(n), f(n+1))$ now has the form

\begin{equation*}
\fzip(g(n), f(n+1))=\begin{cases}
          g(\frac{n}{2}) \quad &\text{n even}\\
          ab^{\frac{n+1}{2}} \quad &\text{n odd}\\
     \end{cases}
\end{equation*}

As in Case 1, this can be rewritten:

\begin{equation*}
\fzip(g(n), f(n+1))=\begin{cases}
          g(\frac{n}{2}) \quad &\text{n even}\\
          bab^{\frac{n-1}{2}} \quad &\text{n odd}\\
     \end{cases}
\end{equation*}

Then we divide each odd block by $b$, which is a reversible FST operation.

\begin{equation*}
h(n)=\begin{cases}
          g(\frac{n}{2}) \quad &\text{n even}\\
          ab^{\frac{n-1}{2}} \quad &\text{n odd}\\
     \end{cases}
\end{equation*}

Now $h(n)$ is just $\fzip(g(n),f(n))$, and therefore $\langle \fzip(f(n),g(n)) \rangle \equiv \langle \fzip(g(n),f(n)) \rangle$. 
\end{proof}

We observe that while it may be possible for functions which are not linear or exponential to have this property (of allowing $\fzip$ to be symmetric) it is worth pointing out that the proof relies on the defining characteristics of linear and exponential functions, namely that for linear functions $f(n+1) = f(n) + a$, and for exponential functions $f(n+1) = bf(n)$. Since our operations on individual blocks are restricted to exactly these operations (addition/subtraction and division/multiplication) this may suggest that these are the only types of functions that could work for this lemma, or at least that expanding this lemma to other types of functions may be difficult. 

We conclude this section by considering a couple of natural questions about $\fzip$: under what circumstances can we compare one $\fzip$ with another, and how can we compare an $\fzip$ of two functions $f$ and $g$ with some third function $h$? For the first question, intuitively we would think that if $f_1(n), g_1(n)$ are above $f_2(n), g_2(n)$ respectively in the transducer hierarchy, then $\fzip(f_1(n), g_1(n))$ should be above $\fzip(f_2(n), g_2(n))$. However, the main difficulty is that while $\fzip$ intertwines two functions, the weight product can't be easily intertwined in the same way. For the second question, we can provide a very weak result which may illustrate the difficulty of this problem.

\begin{lemma}
Let $f,g,h$ be functions from $\mathbb{N}$ to $\mathbb{N}$. Suppose that there exist weights $\bm{\alpha} = (\alpha_0, \alpha_1, ... , \alpha_{n-1})$ and $\bm{\beta} = (\beta_0, \beta_1, ... \beta_{m-1})$ with $f = \bm{\alpha} \otimes h$ and $g = \bm{\beta} \otimes h$. Further suppose that $n = m$ and for all $i$,  $|\alpha_i| = |\beta_i|$. Then $\langle h(n) \rangle \geq \langle \fzip(f(2n), g(2n+1)) \rangle$.
\end{lemma}

\begin{proof}
We will begin by constructing a new weight $\bm{\gamma}$ from $\bm{\alpha}$ and $\bm{\beta}$, and showing that $(\bm{\gamma} \otimes h)(n)$ is equal to $\fzip(f(2n),g(2n+1))$. We have two different cases for constructing $\bm{\gamma}$ depending on if $m$ is even or odd. \\

Case 1: $m$ is even \\

If $m$ is even, simply take $\bm{\alpha}$ and replace each $\alpha_{2i+1}$ with the corresponding $\beta_{2i+1}$ to obtain $\bm{\gamma}$. Then $\bm{\gamma} = (\alpha_0, \beta_1, \alpha_2, \beta_3, ... \beta_{m-1})$. \\

Case 2: $m$ is odd \\

Start as before, by replacing the $\alpha_{2i+1}$'s in $\bm{\alpha}$ with the $\beta_{2i+1}$'s, but this time we will also be replacing the $\beta_{2i+1}$'s in $\bm{\beta}$ with the $\alpha_{2i+1}$ weights from $\bm{\alpha}$. Let's call these new weight tuples $\bm{\alpha'} = (\alpha_0, \beta_1, \alpha_2, \beta_3, ... , \alpha_{m-1})$ and $\bm{\beta'} = (\beta_0, \alpha_1, \beta_2, \alpha_3, ... , \beta_{m-1})$. Now we define $\bm{\gamma}$ to be the weight tuple obtained by concatenating $\bm{\alpha'}$ and $\bm{\beta'}$, i.e. $\bm{\gamma} = ((\alpha_0, \beta_1, \alpha_2, \beta_3, ... , \alpha_{m-1},\beta_0, \alpha_1, \beta_2, \alpha_3, ... , \beta_{m-1})$. \\

Now that we have a definition of $\bm{\gamma}$, we can proceed to proving that $(\bm{\gamma} \otimes h)(n) = \fzip(f(2n),g(2n+1))$. We show this by proving that for even $n$ we have $(\bm{\gamma} \otimes h)(n) = (\bm{\alpha} \otimes h)(n) = f(n)$, and for odd $n$ we have $(\bm{\gamma} \otimes h)(n) = (\bm{\beta} \otimes h)(n) = g(n)$. Let's start by assuming $n$ is even, and from the proof for even $n$ it will be clear that the proof for odd $n$ is identical. \\

If $n$ is even, then by construction of $\bm{\gamma}$ and the definition of the weight product, we have for some $k$,  $(\bm{\gamma} \otimes h)(n) = \alpha_m \cdot S^k(f)$. Similarly, for some $k'$, $(\bm{\alpha} \otimes h)(n) = \alpha_m \cdot S^{k'}(f)$. Letting $L = ||\bm{\gamma}||$, $m$ is equal to $n$ mod L (intuitively, $m$ is the number of times we go through all of the weights in $\bm{\gamma}$). We will show that $k = k'$. \\

The key to this proof is the condition that each $\alpha_i$ has the same length as the corresponding $\beta_i$. Because of this condition, whenever we apply the recursive part of the weight product definition, we are shifting $f$ by the same amount whether we are applying the weights from $\bm{\gamma}$ (which are half $\alpha_i$'s and half $\beta_i$'s) or weights from $\bm{\alpha}$ only. More rigorously, we can see this by calculating the exact values of $k$ and $k'$. 

We have that $k$ is equal to $(n-m)\frac{L}{m} + \sum_{i=0}^{m-1} (|\gamma_i| - 1)$, where $\gamma_i$ is the $i$th weight in $\bm{\gamma}$ ($\gamma_i$ alternates between $\alpha_i$ and $\beta_i$). Similarly, $k' =  (n-m)\frac{L}{m} + \sum_{i=0}^{m-1} (|\alpha_i| - 1)$. From the length condition on the weights in $\bm{\alpha}$ and $\bm{\beta}$, $|\gamma_i| = |\alpha_i|$ for all $i$, and therefore $k = k'$. This means that for even $n$, $(\bm{\gamma} \otimes h)(n) = (\bm{\alpha} \otimes h)(n)$, and by definition of $\bm{\alpha}$, $(\bm{\alpha} \otimes h)(n) = f(n)$. Therefore $(\bm{\gamma} \otimes h)(n) = f(n)$ for even $n$. The case where $n$ is odd proceeds mostly identically, replacing $\bm{\alpha}$ with $\bm{\beta}$ (and $f$ with $g$) where appropriate. So now we have shown that $(\bm{\gamma} \otimes h)(n) = f(n)$ for even $n$ and $(\bm{\gamma} \otimes h)(n) = g(n)$ for odd $n$, which is exactly the definition of $\fzip(f(2n),g(2n+1))$. Therefore $(\bm{\gamma} \otimes h)(n) = \fzip(f(2n), g(2n+1))$ and thus $\langle h(n) \rangle \geq \langle \fzip(f(2n), g(2n+1)) \rangle$.
\end{proof}

\section{A Diamond Structure in the Transducer Hierarchy}

Now we proceed to our main result. We need two more lemmas before moving on to the main proof.

\begin{lemma}

For all polynomials $f$ of degree 2, we can find a weight $\alpha$ of length 2 and integers k, m such that $\alpha \otimes S^k(n^2) = S^mf(n)$.

\end{lemma}

\begin{proof}

Let $f(n) = an^2 + bn + c$. We assume for simplicity that $c = 0$, since the constant term is irrelevant for transducer degrees. We can also assume that $b > 0$, since if not we can simply shift $f$ (i.e. choose a positive $m$) until this holds. Then we claim that the weight $\alpha = \frac{1}{4} \langle a - b + a \lfloor \frac{b}{a} \rfloor, b - a \lfloor \frac{b}{a} \rfloor, -(2b\lfloor \frac{b}{a} \rfloor + b - a\lfloor \frac{b}{a} \rfloor^2) \rangle$, and the integers $k = \lfloor \frac{b}{a} \rfloor, m = 0$, prove the lemma.

Indeed, this is simply a matter of verifying this computationally. From the definition of a weight product, since $\alpha$ only contains one weight, we have \\

\begin{align*}
(\alpha \otimes S^{\lfloor \frac{b}{a} \rfloor}n^2)(n) &= \frac{1}{4}(a - b + a \lfloor \frac{b}{a} \rfloor)(2n+\lfloor \frac{b}{a} \rfloor)^2 \\
&\quad + \frac{1}{4}(b - a \lfloor \frac{b}{a} \rfloor)(2n+1+\lfloor \frac{b}{a} \rfloor)^2 \\
&\quad -\frac{1}{4}(2b\lfloor \frac{b}{a} \rfloor + b - a\lfloor \frac{b}{a} \rfloor^2) \\
&= \frac{1}{4}(a - b + a \lfloor \frac{b}{a} \rfloor)(4n^2 + 4n\lfloor \frac{b}{a} \rfloor + \lfloor \frac{b}{a} \rfloor^2) \\
&\quad + \frac{1}{4}(b - a \lfloor \frac{b}{a} \rfloor)(4n^2 +4n + 4n\lfloor \frac{b}{a} \rfloor +2\lfloor \frac{b}{a} \rfloor \\
&\quad + \lfloor \frac{b}{a} \rfloor^2 + 1) - \frac{1}{4}(2b\lfloor \frac{b}{a} \rfloor + b - a\lfloor \frac{b}{a} \rfloor^2) \\
&= \frac{1}{4}(4an^2 + 4an\lfloor \frac{b}{a} \rfloor + a\lfloor \frac{b}{a} \rfloor^2 - 4bn^2 - 4bn\lfloor \frac{b}{a} \rfloor - b\lfloor \frac{b}{a} \rfloor^2 \\
&\quad + 4an^2\lfloor \frac{b}{a} \rfloor + 4an\lfloor \frac{b}{a} \rfloor^2 + a\lfloor \frac{b}{a} \rfloor^3) \\
&\quad + \frac{1}{4}(4bn^2 + 4bn + 4bn\lfloor \frac{b}{a} \rfloor + 2b\lfloor \frac{b}{a} \rfloor + b + b\lfloor \frac{b}{a} \rfloor^2 \\
&\quad - 4an^2\lfloor \frac{b}{a} \rfloor - 4an\lfloor \frac{b}{a} \rfloor - 4an\lfloor \frac{b}{a} \rfloor^2 - 2a\lfloor \frac{b}{a} \rfloor^2 \\
&\quad -a\lfloor \frac{b}{a} \rfloor - a\lfloor \frac{b}{a} \rfloor^3) -\frac{1}{4}(2b\lfloor \frac{b}{a} \rfloor + b - a\lfloor \frac{b}{a} \rfloor^2) \\
&= \frac{1}{4}(4an^2 + 4bn) + \frac{1}{4}(2b\lfloor \frac{b}{a} \rfloor + b - a\lfloor \frac{b}{a} \rfloor^2) \\
&\quad - \frac{1}{4}(2b\lfloor \frac{b}{a} \rfloor + b - a\lfloor \frac{b}{a} \rfloor^2) \\
&= an^2 + bn
\end{align*}

However, we do need to be careful here, and make sure to check that $\alpha$ is indeed a valid weight. So we need its entries to both be nonnegative.

Because $\lfloor \frac{b}{a} \rfloor \leq \frac{b}{a}$, $b - a \lfloor \frac{b}{a} \rfloor \geq b - a \frac{b}{a} = 0$. So the second entry is nonnegative. Similarly, since $\lfloor \frac{b}{a} \rfloor \geq \frac{b}{a} - 1$, $\langle a - b + a \lfloor \frac{b}{a} \rfloor \geq \langle a - b + a(\frac{b}{a} - 1) = a - b +b - a = 0$ and the first entry is also nonnegative. So $\alpha$ is a valid weight, and $\alpha \otimes S^k(n^2) = f(n)$.
\end{proof}

Now we prove the second lemma, which gives us essentially the inverse statement of the previous lemma. The basic computations behind this lemma were modified from Theorem 5.2 in \cite{2} to better suit the purposes of this paper.

\begin{lemma}
Let $f$ be a quadratic function of the form $a(n+1)^2 + b(n+1)$, with $2a > b > 0$. Then there is a weight $\alpha$ of length 2 such that $\alpha \otimes f$ = $(n+1)^2$.
\end{lemma}

\begin{proof}
We claim that the weight $\alpha = \frac{1}{8a^2}(b, 2a-b, b^2+ab+6a^2+1)$ satisfies $(\alpha \otimes f)(n) = (n+1)^2$. This can be verified computationally, by simplifying the expression $\frac{1}{8a^2}(bf(2n) + (2a-b)f(2n+1))$, and noting that since $2a > b > 0$, both $b$ and $2a-b$ are positive, making $\alpha$ a valid weight.
\end{proof}

\begin{theorem}

The degree $\langle \fzip(n,n^2) \rangle$ is strictly greater than $\langle n \rangle$, and there are no intermediate degrees.

\end{theorem}

\begin{proof}

The fact that $\langle \fzip(n,n^2) \rangle$ is strictly greater than $\langle n \rangle$ is trivial, since if $n$ could be transduced into $\fzip(n,n^2)$, it could also be transduced into $n^2$, which has been shown to be impossible. Proving that there are no intermediate degrees is the nontrivial part.

We set out to prove this statement by assuming that there is some intermediate degree. Since $\fzip(n,n^2)$ is a piecewise polynomial function, the degree of anything below it is equivalent to $\langle g \rangle$ for some (piecewise) polynomial function $g$, and in particular $g$ is a weight product of $\fzip(n,n^2)$. Because $\fzip(n,n^2)$ is a piecewise polynomial function, any weight product will also be a piecewise polynomial function, and the pieces will all be linear or quadratic polynomials. (We can remove constant functions without any loss of degree). So then there are three cases: \\

Case 1: $g$ has only quadratic polynomials as its pieces. \\

Case 2: $g$ has only linear polynomials as its pieces. \\

Case 3: $g$ has both linear and quadratic polynomials as its pieces. \\ 

Recall that by ``pieces", we mean that since $g$ can be written as a function which is defined piecewise by N different functions, we call these functions the ``pieces" of $g$.

Now we consider each case. For Case 1, $g$ cannot be transduced into $n$ if each piece is quadratic, since the block size grows too fast to allow this. For Case 2, where $g$ is a piecewise linear function, suppose $g$ is equal to: 

\begin{equation*}
g(n)=\begin{cases}
          a_1n+b_1 \quad &\text{n $\equiv$ 0modN}\\
          a_2n+b_2 \quad &\text{n $\equiv$ 1modN}\\
	\vdots \\
	a_Nn+b_N \quad &\text{n $\equiv$ N-1modN}\\
     \end{cases}
\end{equation*}

Then taking the product of the weight $\bm{\alpha}$ $= (\alpha_1, \alpha_2, \cdots, \alpha_N)$, where $\alpha_i = (a_i, b_i)$, with $n$ shows that in this case $\langle n \rangle \geq \langle g \rangle$, and since $g$ is clearly transducible to $n$ they must be of equal degree.

Now we have only Case 3, where $g$ has both linear and quadratic pieces. We will show that $g$ can be transduced back into $\fzip(n,n^2)$, and this will complete the proof.

First, we can assume that the first piece of $g$ is linear by simply deleting blocks from the beginning of $\langle g \rangle$.

We can then combine all of the other pieces of $g$ via a weight of the form $\langle (1,0),(1,1,\cdots,1,0) \rangle$, which preserves the first piece and sums up all of the other pieces. So now $g$ has the form 

\begin{equation*}
g_1(n)=\begin{cases}
          an+b \quad &\text{n even}\\
          An^2+Bn+C \quad &\text{n odd}\\
     \end{cases}
\end{equation*}

We can easily construct a transducer to subtract $b$ from only the even-numbered blocks and $C$ from the odd-numbered blocks, and also a transducer to divide even-numbered blocks by $a$. Then $g_2$ has the form

\begin{equation*}
g_2(n)=\begin{cases}
          n \quad &\text{n even}\\
          An^2+Bn \quad &\text{n odd}\\
     \end{cases}
\end{equation*}

Remove the first two blocks to shift $g_2$ to the right by 2, and then subtract 2 from all even blocks to obtain:

\begin{equation*}
g_3(n)=\begin{cases}
          n \quad &\text{n even}\\
          A(n+2)^2+B(n+2) \quad &\text{n odd}\\
     \end{cases}
\end{equation*}

Note that we can rewrite $A(n+2)^2+B(n+2)$ as $An^2+4An+4+Bn+2B = An^2 + (4A+B)n+2B+4 = An^2 + B'n + C'$, where $B' = 4A+B$ and $C' = 2B+4$. We can repeat this process indefinitely, and clear away the $C'$'s at the end by repeating the step from $g_1(n)$ to $g_2(n)$. Note that we can also reverse this process to get an expression of the form $A(n-2)^2+B(n-2)$, and so we can also arbitrarily lower the value of $B$. Since we're adding or subtracting multiples of $4A$ to $B$, we can return to the form of $g_2(n)$ with a guarantee that $2A \leq B < 6A$.

\begin{equation*}
g_4(n)=\begin{cases}
          n \quad &\text{n even}\\
          An^2+Bn \quad &\text{n odd}\\
     \end{cases}
\end{equation*}

Now let $h(n) = An^2+Bn$. Define a new function $f(n) = h(2n+1)$. Rewriting $f(n)$ in terms of $n+1$ gives $f(n) = 4A(n+1)^2+(-4A+2B)(n+1)+D$, where $D$ is the constant term. We have already seen that the constant term is easily eliminated, so its exact value is irrelevant.

We now want to apply Lemma 4.2 to $f(n)$. Since $2A \leq B < 6A$, the coefficient $(-4A+2B)$ of $(n+1)$ satisfies $0 \leq -4A+2B < 8A$. Note that $8A$ is two times the coefficient of $(n+1)^2$, and therefore Lemma 4.2 can be applied unless $-4A+2B=0$. But that would be the case $B = 2A$,  and then $g_4(n)$ would be easily transducable to $g_5(n)$ by completing the square and dividing by 4A. So we can assume that $B \neq 2A$. 

Therefore $f(n)$ can be transduced into $(n+1)^2$ via a weight product with only one weight of length 2. Let's call this weight $\bm{\alpha}$ = $(\alpha_1, \alpha_2, \beta)$ (The lemma tells us exactly what these are, but it's not important). Now we can take the product of the modified weight $\bm{\alpha^*} = \langle (\frac{1}{2},0), (\alpha_1,0,\alpha_2, \beta) \rangle$ with $g_4$. This weight allows us to ignore the alternating nature of $g_4$ and only transduce the $An^2+Bn$ piece. Since this part of the proof is very technical, we include some tables and diagrams to clarify the situation. First, a table clarifying exactly what $g_4(n)$ looks like.

\begin{center}
  \begin{tblr}{hlines, vlines, columns={co=1}} 
    n & 0 & 1 & 2 & 3 & 4 & 5 & 6 & 7 & 8   \\
    $g_4(n)$ & 0 & h(1) = f(0) & 2 & h(3) = f(1) & 4 & h(5) = f(2) & 6 & h(7) = f(3) & 8   \\
  \end{tblr}
\end{center} 

Now we examine the weight product of $\bm{\alpha^*}$ with $g_4(n)$.

\begin{enumerate}
\item[] 
\begin{tikzpicture}[grow=up]
\Tree [.{$(\bm{\alpha^*} \otimes g_4)(0) = \frac{1}{2}g_4(0) + 0 = 0$} [.{$\times \frac{1}{2}$} $g_4(0)$ ] ]
\end{tikzpicture}
\hskip 0.1in
\begin{tikzpicture}[grow=up]
\Tree [.{$(\bm{\alpha^*} \otimes g_4)(1) = \alpha_1g_4(1) + 0 + \alpha_2g_4(3) + \beta$} [.{$\times \alpha_2$} $g_4(3)$ ] [.{$\times 0$} $g_4(2)$ ] [.{$\times \alpha_1$} $g_4(1)$ ] ]
\end{tikzpicture}
\end{enumerate}

\begin{enumerate}
\item[] 
\begin{tikzpicture}[grow=up]
\Tree [.{$(\bm{\alpha^*} \otimes g_4)(2) = \frac{1}{2}g_4(4) + 0 = 2$} [.{$\times \frac{1}{2}$} $g_4(4)$ ] ]
\end{tikzpicture}
\hskip 0.1in
\begin{tikzpicture}[grow=up]
\Tree [.{$(\bm{\alpha^*} \otimes g_4)(3) = \alpha_1g_4(5) + 0 + \alpha_2g_4(7) + \beta$} [.{$\times \alpha_2$} $g_4(7)$ ] [.{$\times 0$} $g_4(6)$ ] [.{$\times \alpha_1$} $g_4(5)$ ] ]
\end{tikzpicture}
\end{enumerate}

From the above diagram giving the first four values of $\bm{\alpha^*} \otimes g_4$, it should be believable that the values on even $n$ are exactly $n$, as was the case for $g_4$ itself. The values for odd $n$ are a bit more complicated. The general pattern is that for odd $n$, $(\bm{\alpha^*} \otimes g_4)(n) = \alpha_1g_4(2n-1)+\alpha_2g_4(2n+1)+\beta$. However, by definition of the function $h(n)$ from earlier, this is equivalent to $(\bm{\alpha^*} \otimes g_4)(n) = \alpha_1h(2n-1) + \alpha_2h(2n+1)+\beta$. But if we redefine this in terms of $f(n)$, we obtain $(\bm{\alpha^*} \otimes g_4)(n) = \alpha_1f(n-1) + \alpha_2f(n)+\beta$. 

Now the right hand side of this equation is in fact equal to $(\bm{\alpha} \otimes f)(\frac{n-1}{2})$. But $\bm{\alpha}$ was chosen to be the weight such that $(\bm{\alpha} \otimes f)(n) = (n+1)^2$. Therefore, for odd $n$, $(\bm{\alpha^*} \otimes g_4)(n) = (\frac{n-1}{2}+1)^2 = (\frac{n+1}{2})^2$, and renaming $(\bm{\alpha^*} \otimes g_4)(n)$ to $g_5(n)$ gives:   

\begin{equation*}
g_5(n)=\begin{cases}
          n \quad &\text{n even}\\
          (\frac{n+1}{2})^2 \quad &\text{n odd}\\
     \end{cases}
\end{equation*}

Now we can add two blocks to shift $g_5$ to the left by 2, then add 2 and divide by 2 on all even numbered blocks to preserve $n$:

\begin{equation*}
g_6(n)=\begin{cases}
          \frac{n}{2} \quad &\text{n even}\\
          (\frac{n-1}{2})^2 \quad &\text{n odd}\\
     \end{cases}
\end{equation*}

which is exactly $\fzip(n,n^2)$. Since each $g_i$ was formed by transducing $g_{i-1}$, we have that $\langle g \rangle \geq \langle \fzip(n,n^2) \rangle$, and thus $\langle g \rangle = \langle \fzip(n,n^2) \rangle$. \\

Therefore there are no intermediate degrees between $\fzip(n,n^2)$ and $n$.
\end{proof}

Now we proceed to prove the same result for $n^2$.

\newpage

\begin{theorem}

The degree $\langle \fzip(n,n^2) \rangle$ is strictly greater than $\langle n^2 \rangle$, and there are no intermediate degrees.

\end{theorem}

\begin{proof}

The proof will proceed in a similar manner to the previous theorem. Again the fact that $\langle \fzip(n,n^2) \rangle$ is strictly greater than $\langle n^2 \rangle$ is trivial, since otherwise $n^2$ could be transduced into $n$. So we only need to prove that there are no intermediate degrees.

Letting $g$ be a potential intermediate degree, we have the same three cases as before: \\

Case 1: $g$ has only quadratic polynomials as its pieces. \\

Case 2: $g$ has only linear polynomials as its pieces. \\

Case 3: $g$ has both linear and quadratic polynomials as its pieces. \\ 

Case 2 is not possible because we showed in the previous theorem that such a $g$ would be the same degree as $n$. For Case 3, the previous theorem proved that $\langle g \rangle$ would be equal to $\langle \fzip(n,n^2) \rangle$. So we need to turn our attention to Case 1. For this case, we will prove that $g$ is the same degree as $n^2$, and this will complete the proof.

If $g$ has only quadratic polynomials as its pieces, it has the form:

\begin{equation*}
g(n)=\begin{cases}
          a_1n^2+b_1n+c_1 \quad &\text{n $\equiv$ 0modN}\\
          a_2n^2+b_2n+c_2 \quad &\text{n $\equiv$ 1modN}\\
	\vdots \\
	a_Nn^2+b_Nn+c_N \quad &\text{n $\equiv$ N-1modN}\\
     \end{cases}
\end{equation*}

We can remove all of the constant terms:

\begin{equation*}
g(n)=\begin{cases}
          a_1n^2+b_1n \quad &\text{n $\equiv$ 0modN}\\
          a_2n^2+b_2n \quad &\text{n $\equiv$ 1modN}\\
	\vdots \\
	a_Nn^2+b_Nn \quad &\text{n $\equiv$ N-1modN}\\
     \end{cases}
\end{equation*}

Now using the result of Lemma 4.1, let $\bm{\alpha} = (\alpha_1, \alpha_2, \cdots, \alpha_N)$ where $\alpha_i =  \frac{1}{4} \langle a_i - b_i + a_i \lfloor \frac{b_i}{a_i} \rfloor, b_i - a_i \lfloor \frac{b_i}{a_i} \rfloor, -(2b_i\lfloor \frac{b_i}{a_i} \rfloor + b_i - a_i\lfloor \frac{b_i}{a_i} \rfloor^2) \rangle$. If we take the product of this weight with $n^2$, then we obtain $g$. To see that this is the case, we will do another example weight product calculation, for the case $N = 2$, and use a similar diagram as the previous theorem to illustrate our point. 

Let the weight $\bm{\alpha} = (\alpha_1, \alpha_2)$, where $\alpha_1 = (r_1, r_2, s)$ and $\alpha_2 = (u_1, u_2, v)$. The exact values of these are as above. The key point is that $\alpha_1, \alpha_2$ satisfy the equations $(\alpha_1 \otimes f)(n) = a_1n^2+b_1n$ and $(\alpha_2 \otimes f)(n) = a_2n^2+b_2n$, where $f(n) = n^2$. Now if we calculate the first four values of $\bm{\alpha} \otimes n^2$, we get the following:

\begin{enumerate}
\item[] 
\begin{tikzpicture}[grow=up]
\Tree [.{$(\bm{\alpha} \otimes f)(0) = r_1f(0)+r_2f(1)+s$} [.{$\times r_2$} $f(1)$ ] [.{$\times r_1$} $f(0)$ ] ]
\end{tikzpicture}
\hskip 0.1in
\begin{tikzpicture}[grow=up]
\Tree [.{$(\bm{\alpha} \otimes f)(1) = u_1f(2)+u_2f(3)+v$} [.{$\times u_2$} $f(3)$ ] [.{$\times u_1$} $f(2)$ ] ]
\end{tikzpicture}
\end{enumerate}

\begin{enumerate}
\item[] 
\begin{tikzpicture}[grow=up]
\Tree [.{$(\bm{\alpha} \otimes f)(2) = r_1f(4)+r_2f(5)+s$} [.{$\times r_2$} $f(5)$ ] [.{$\times r_1$} $f(4)$ ] ]
\end{tikzpicture}
\hskip 0.1in
\begin{tikzpicture}[grow=up]
\Tree [.{$(\bm{\alpha} \otimes f)(3) = u_1f(6)+u_2f(7)+v$} [.{$\times u_2$} $f(7)$ ] [.{$\times u_1$} $f(6)$ ] ]
\end{tikzpicture}
\end{enumerate}

Looking at this example, we can now see that in fact, $(\bm{\alpha} \otimes f)(n) = (\alpha_1 \otimes f)(n)$ for even $n$ and $(\bm{\alpha} \otimes f)(n) = (\alpha_2 \otimes f)(n)$ for odd $n$. But because of the way that $\alpha_1$ and $\alpha_2$ were defined, this means that $(\bm{\alpha} \otimes f)(n) = a_1n^2+b_1n$ for even $n$ and $(\bm{\alpha} \otimes f)(n) = a_2n^2+b_2n$ for odd $n$, which is exactly the definition of $g$ in the case $N = 2$. So $\bm{\alpha} \otimes n^2 = g$. The general case proceeds effectively the same way. The key property of the weight tuple $\bm{\alpha}$ for this proof is that all of the weights in the tuple are the same length, which allows $n^2$ to be turned into a piecewise polynomial like $g$ with each weight corresponding exactly to one piece. The $N = 2$ case could also be thought of in terms of Lemma 3.6, with the general case being provable via an extension of that lemma.

Therefore $\langle n^2 \rangle \geq \langle g \rangle$. But $g$ was a degree between $\fzip(n,n^2)$ and $n^2$, and therefore $\langle n^2 \rangle = \langle g \rangle$. Since this was the last case for $g$, there are no intermediate degrees between $\fzip(n,n^2)$ and $n^2$.
\end{proof}

\section{Conclusion}

We have shown that $\fzip(n,n^2)$ lies strictly above both $n$ and $n^2$, with no intermediate degrees between them. From earlier results, we know that the degrees of both $n$ and $n^2$ are atoms, that is, there is nothing between them and the bottom degree $\textbf{0}$. Therefore $\fzip(n,n^2)$ forms a diamond structure with $n, n^2$ and $\textbf{0}$, and this is the first such structure that has been found. We also note that by the case analysis for potential transducts of $\fzip(n,n^2)$ in the previous section, $n$ and $n^2$ are the only degrees below $\fzip(n,n^2)$. This result sheds more light on the structure of the transducer hierarchy, and also raises some further questions about the potential use of $\fzip$ to find new results. We state a few of these questions here:

\begin{enumerate}
\item We know that $\langle \fzip(f,f) \rangle \geq \langle f \rangle$. In general, is this inequality strict, or can $\langle \fzip(f,f) \rangle = \langle f \rangle$ for more than just linear or quadratic $f$'s?

\item What can we say about $\langle \fzip(f,g) \rangle$ when $f,g$ are both cubic polynomials?

\item For some $f,g$ is it possible to find degrees between $f$ and $\fzip(f,g)$?

\item If $\langle f_1 \rangle \geq \langle f_2 \rangle$ and $\langle g_1 \rangle \geq \langle g_2 \rangle$, then is $\langle \fzip(f_1,g_1) \rangle \geq \langle \fzip(f_2,g_2) \rangle$? 
\end{enumerate}

\bibliographystyle{abbrvnat}
\bibliography{bibliography}

\end{document}